\documentclass[english,preprint,10pt]{elsarticle}

\usepackage{graphicx}
\usepackage{amssymb}
\usepackage{amsthm}
\usepackage{amsmath}
\usepackage{babel}

\newtheorem{theorem}{Theorem}[section]
\newtheorem{lemma}[theorem]{Lemma}

\newtheorem{corollary}[theorem]{Corollary}
\newenvironment{definition}[1][Definition]{\begin{trivlist}
\item[\hskip \labelsep {\bfseries #1}]}{\end{trivlist}}

\begin{document}

\begin{frontmatter}

\title{On the dynamic compressibility of sets}

\author[ufl]{Karthik S. Gurumoorthy\corref{cor1}}
\ead{sgk@ufl.edu}

\cortext[cor1]{Corresponding Author\\
Address:\\
E301, CSE Building, University of Florida,\\
P.O. Box 116120, Gainesville, FL 32611-6120, USA.\\
Ph: 001-352-392-1200, Fax: 001-352-392-1220
}
\address[ufl]{Department of Computer and Information Science and Engineering, University of Florida, Gainesville, Florida, USA}

\begin{abstract}
We define a new notion of compressibility of a set of numbers through the dynamics of a polynomial function. We provide approaches to solve the problem by reducing it to the multi-criteria traveling salesman problem through a series of transformations. We then establish computational complexity results by giving some NP-completeness proofs. We also discuss about a notion of $\epsilon$ K-compressibility of a set, with regard to lossy compression and deduce the necessary condition for the given set to be $\epsilon$ K-compressible. Finally, we conclude by providing a list of open problems solutions to which could extend the applicability the our technique. 
\end{abstract}

\begin{keyword}
Compression, Dynamical systems, Computational complexity, Assignment problem, Sparse-representation
\end{keyword}

\end{frontmatter}

\section{Introduction\label{introduction}}
To the best of our knowledge, all the compression techniques existing in the literature are \emph{static} in nature. By static, we refer to the methods which look for repetitive patterns in the input data and perform compression by encoding the data as codewords. The encoding process typically involves assigning shorter codewords to the most repeated patterns and longer codewords for the seldom ones. A look-up table stores these pattern-codeword matchings and is transmitted along with the encoded data. Decoding process comprises of searching the look-up table and retrieving the pattern corresponding to the received codeword. The number and the length of the codewords are determined by the error correcting capability of the associated technique. The most widely used static compression technique is the Huffman coding--an entropy encoding algorithm used for lossless data compression \cite{Huffman52}. Number theoretic methods like arithmetic encoding \cite{Langdon84} and range encoding \cite{Martin79} schemes have also gained popularity. These entropy based techniques are known to compress messages close to the theoretical limit computed via entropy. Other lossless compression techniques like the Lempel-Ziv-Welch (LZW) compression \cite{Lempel78, Welch84} and its many variants are also widely known in literature. They build the dictionary (codewords) by scanning through the input string for successively longer substrings until they find one that is not in the dictionary.

In many applications, one may encounter scenarios where one would like to compress a \emph{set} as a whole, where the chronological order in which the elements appear in the set are irrelevant. It is acceptable as long as the elements constituting the set are retrieved in some order \emph{regardless} of the order in which they originally featured in the set. Such cases occur quite often especially when transmitting large databases over the internet. If we associate a number to each tuple representing (say) an employee record, the order in which these tuples (numbers) are received at the recipient end may not have any significance. Even if the positional information is needed, bits indicating the position number can be padded to the tuple and then sent. Another case where the positional information of data is of very little relevance is in building a cumulative distribution from empirical data. These scenarios only demand that the data is transmitted as a set and not necessarily as a sequence.

The existing techniques in the literature seems to be blinded to this distinction between the compression a set and that of a sequence. All the aforementioned static compression technqiues could very well be used to compressing a set, as these methods only focusses on the frequency of the repetitive patterns in the input data and hence are invariant to the order in which these patterns appear. Unfortunately, the strength of these techniques is also its weakness. Namely, their compression efficiency is also invariant to an arbitrary permutation of the set and hence remains the same for both a set and for a sequence obtained by imposing an order on the set say by a permutation of its elements. Other lossy dictionary based techniques like the overcomplete DCT, KSVD which tries to compress a given collection of \emph{vectors} by representing them by a corresponding sparse vector built using the dictionary, also relies on an ordering imposed on the set, i.e, one need to explicity identify the first component of the vector, its second component etc and the compression efficiency will depend on the order in which one identifies these components. 

In this work, we introduce a new approach to the compression of a set of numbers through the dynamics of a polynomial function. The fundamental question that we try to answer is: \emph{``What is the computational complexity of determining the permutation of the set for which it is maximally compressible?''}. The maximal compression of a set is connected with the degree of the polynomial function that is used for compression, a notion that will become clearer as we proceed. Before we delve into the intricacies of our technique, let us consider a very simple example. Let us assume that we wish to compactly represent the set $S$ consisting of first 100 natural numbers, namely $S = \{1,2,3,\ldots,99,100\}$. Static compression techniques may first encode this set in binary and then look for patterns in the long sequence of binary digits. Codewords are then assigned to each of the patterns to obtain the encoded data. But given this set $S$, we immediately notice a global structure which these static compression techniques are oblivious of, namely the elements of the set $S$ can be produced in succession by adding $1$ to the previous element. In the terminology of dynamical system, this set $S$ can be produced starting from $1$ and by repeated composition of the function $f(x) = x+1$ with itself. Each composition of the function produces one element of the set. Hence the most succinct representation of this set is the starting element $1$ and the coefficients of the iterating function $f$.

The set defined above is na$\ddot{i}$ve and the structure it exhibits is simple and conspicuous. Instead consider the following set of 10 numbers, 
\begin{equation*}
\tilde{S} = \{0.0016,0.3016,0.0990,0.9178,0.15,0.51,0.9996,0.3567,0.0064,0.0254\}
\end{equation*}
which appears to be totally random. But when rearranged, the elements of the set form the orbit of the first 10 iterations of the quadratic map $f(x) = 4x(1-x)$ starting from $0.15$. Hence the arbitrary looking set $\tilde{S}$ can actually be very compactly represented by encoding only the four numbers--the triplet $(-4,4,0)$ corresponding to the coefficients of the quadratic map and the starting element $0.15$. Numerous other examples can be constructed where a random looking set can actually be produced by iterating a smaller degree polynomial and hence can be efficiently represented. This paper provides ways to test whether a given set imports such structures. 

Let us now begin with the formal definition of when we consider a set to be intrinsically compressible.
\begin{definition}
A set S consisting of $N$ distinct numbers is \emph{K-compressible} if there is exist a polynomial $f$ of degree utmost $K$ and a number $x_0 \in S$, such that $S = \{x_0, f(x_0),\ldots,f^{N-1}(x_0)\}$.
\end{definition}
Here $f^i$ denotes the composition of $f$ with itself $i$ times. We would like to emphasize that the above definition is based on equivalence of sets and not on sequence, i.e., the order in which the elements of $S$ occur in the orbit of $f$ starting at $x_0$ is irrelevant. The immediate question that follows is: \emph{``How to determine whether a given set is K-compressible?''} 

If the sequence information is known, i.e, the sequence in which the elements of $S$ occur in the orbit of $f$ are known a-priori, then it is a fairly easy task to find out whether the given set is K-compressible. Without loss of generality, let $x_1,x_2,\ldots,x_N$ be the known sequence. By the \emph{unisolvence} theorem, there exist a unique polynomial $f$ of degree utmost $K$ with the property that $f(x_i) = x_{i+1}, 1 \leq i \leq K+1 $. The coefficients of $f$ can easily be computed by inverting the corresponding Vandermonde matrix (many faster algorithms like Newton and Lagrange interpolations do not require to invert the Vandermonde matrix). Once the coefficients of $f$ are determined, it is a trivial task to check whether $f(x_j) = x_{j+1}, K+2 \leq j \leq N-1 $.

The problem is interesting only the sequence information is unknown and only when $K < N-2$. When $ K \geq N-2$, by the above argument we see that for each of the possible $N!$ sequences, there exist a unique $f$ of degree utmost $N-2$ which produces the numbers in that sequence. Hence the set is trivially K-compressible. When $K < N-2$, is it the case that all the $\binom {N} {K+2} (K+2)!$ possible sequences need to considered and verified separately before reaching a conclusion? The remaining sections of the paper answers this one question. 

The paper is organized as follows. Section (\ref{property}) states and proves a property of K-compressible sets. In section (\ref{EqvTrans}) we give approaches to check for K-compressibility by reducing it to other well known problems through a series of equivalent transformations. In section (\ref{ComputationalComplexity}) we establish computational complexity results and in section (\ref{epsilonCompressibility}) we discuss the notion of $\epsilon$ K-compressibility with regard to lossy compression and deduce the necessary condition for the given set to be $\epsilon$ K-compressible. We conclude in section (\ref{conclusion}) by providing a list of open problems, some of which are interesting from a theoretical perspective.

\section{Property of K-compressible sets}
\label{property}
The following is a useful property of K-compressible sets.
\begin{lemma}
\label{affineTrans}
Every affine transformation of a K-compressible set is K-compressible.
\end{lemma}
\begin{proof}
Let $S = \{x_1,x_2,\ldots,x_N\}$ be a K-compressible set. Let the $K^{th}$ degree polynomial $f(x)= \sum_{k=0}^K a_k x^k$ be the compressing function. Without loss of generality let $f(x_i) = x_{i+1}, 1\leq i \leq N$. Let $y_i = \lambda x_i+ \gamma, 1\leq i\leq N, \lambda,\gamma \in \mathbb{R}$. Denote $\tilde{S} = \{y_1,y_2,\ldots,y_N\}$. We now show the existence of $g \in \prod_K$ ($\prod_K$ denotes the vector spaces of polynomials of degree at most K) with the property that $g(y_i) = y_{i+1}$ and hence prove that $\tilde{S}$ is K-compressible. In what follows, let $g(y) = \sum_{k=0}^K b_k y^k$.\\\\
case(i): Let $\gamma = 0$. The only interesting case is $\lambda \not= 0$. Let $b_k = \frac{a_k}{\lambda^{k-1}}$. Then 
\begin{equation*}
g(y_j) = \sum_{k=0}^K \lambda a_k x_j^k = \lambda x_{j+1} = y_{j+1}, 1 \leq j \leq N-1.
\end{equation*}\\\\
case (ii): Let $\gamma = \lambda = 1$. Then for $1 \leq j \leq N-1$,
\begin{eqnarray*}
g(y_j) &=& \sum_{k=0}^K b_k(x_j+1)^k = \sum_{k=0}^K \sum_{l=0}^k b_k \binom{k}{l}x_j^l \nonumber\\
       &=& \sum_{k=0}^K x_j^k \left(\sum_{l=k}^K b_l \binom{l}{k}\right). \nonumber \\
\end{eqnarray*}
Define $a_k = \sum_{l=k}^K b_l \binom{l}{k}, 1\leq k \leq K$, $a_0+1= \sum_{l=0}^K b_l$ and solve the linear system of equations to obtain $b_0,b_1,\ldots,b_K$ as a function of $a_0,a_1,\ldots,a_K$. Hence,
\begin{equation*}
g(y_j) = \sum_{k=0}^K a_k x_j^k + 1 = x_{j+1}+1 = y_{j+1}.
\end{equation*}\\\\
case (iii): Let $\lambda = 1, \gamma \not=0$. Let $z_j = \frac{x_j} {\gamma}$. Then $y_j = \gamma (z_j+1)$. By case(i), $\exists f_1\in \prod_K$ with $f_1(z_j) = z_{j+1}$. Replacing $x_j$ with $z_j$ in case(ii), $\exists f_2 \in \prod_K$with $f_2(z_j+1) = z_{j+1}+1$. Replacing $x_j$ with $z_j+1$ in case (i), $\exists g \in \prod_K$ with 
\begin{equation*}
g(\gamma(z_j+1)= y_j) = \gamma(z_{j+1}+1) = y_{j+1}.
\end{equation*}\\\\
case (iv): For the general case, let $z_j = \lambda x_j$. Then $y_j = z_j + \gamma$. By case (i),$\exists f_1\in \prod_K$ with $f_1(z_j) = z_{j+1}$. Replacing $x_j$ with $z_j$ in case (iii), $\exists g \in \prod_K$ with 
\begin{equation*}
g(z_j+\gamma= y_j) = z_{j+1}+\gamma = y_{j+1}
\end{equation*}
which completes the proof.   
\end{proof}
\section{Equivalent transformations}
\label{EqvTrans}
We now provide a set of equivalent transformations of our problem which are useful in establishing computational complexity results and in providing approaches to solve the problem. The K-compressibility problem can be formally restated as follows. Given a set $S$ of $N$ distinct numbers $S = \{x_1,x_2,\ldots,x_N\}$ and an integer $K<N-2$, does $ \exists \sigma \in S_N$ ($S_N$ denotes the symmetric group) and a $f \in \prod_K$ such that $f(x_{\sigma(i)}) = x_{\sigma(i+1)}, 1 \leq i \leq N-1$. We would like to stress on this key point that if a set $S$ is K-compressible, then elements of $S$ should be producible starting at some point $x_0 \in S$ and by repeated composition of $K^{th}$ polynomial $f$. In other words, $f^{i}(x_0) \in S, 0\leq i \leq N-1$ and $f^i(x_0) \not= f^j(x_0)$ if $i \not = j$.

Immediately we observe similarities between the constraints imposed above and the constraints imposed in traveling salesman problem and the Hamiltonian path problem. If we regard each number as a city (or vertex), the problem then is to find the sequence of cities(numbers) such that each city is visited exactly once and a certain other condition is met. We call this constraint the \emph{Hamiltonian path constraint}. In what follows, we show how these disparate looking problems having similar conditions can be tied together.

From the given set $S$, choose a $y \in S$ and assume that $\exists x_0 \in S$ and $f \in \prod_K$ such that $f^{N-1}(x_0) = y$ and satisfying the Hamiltonian path constraint. The problem then reduces to determining whether for at least one choice (out of possible $N$ choices) of $y$, our assumption holds good. Without loss of generality let $x_N = y$. We then need to determine whether there exist numbers $a_0,a_1,\ldots,a_K$ such that the following matrix equation
\[ 
\left[
\begin{matrix}
1&x_1&x_1^2&\dots&x_1^K\\
1&x_2&x_2^2&\dots&x_2^K\\
\vdots&\vdots&\vdots&\vdots&\vdots\\
1&x_{N-1}&x_{N-1}^2&\dots&x_{N-1}^K
\end{matrix}
\right]
\left[
\begin{matrix}
a_0\\a_1\\\vdots\\a_K
\end{matrix}
\right] = \left[
\begin{matrix}
v_1\\v_2\\\vdots\\v_{N-1}
\end{matrix}
\right]
\]
is satisfied along with the Hamiltonian path constraint. The matrix equation can be succinctly represented as $V \vec{a} = \vec{v}$. Here $v_1,v_2,\ldots,v_{N-1} \in S$. Note that the missing element $v_N$ is the starting element $x_0$. For $K < N-2$, the above set of equations form an over-determined system. The matrix $V$ on the left of the equation is commonly referred to as \emph{Vandermonde} matrix and is well known in polynomial interpolation theory.

Using the least squares solution for $\vec{a}$, namely $\vec{a} = (V^TV)^{-1}V^T\vec{v}$, $\vec{v}$ satisfies 
\begin{equation*}
V(V^TV)^{-1}V^T\vec{v} = \vec{v}, 
\end{equation*}
i.e, $\vec{v}$ is the eigen-vector of $V(V^TV)^{-1}V^T$ with eigen-value $1$.

Let $V = \left[Q_1 Q_2 \right]\left[\begin{matrix}R_1\\0\end{matrix}\right]$ denote the $QR$ decomposition of $V$, where $ Q = [Q_1 Q_2]$ is an orthogonal matrix and $R_1$ is $K+1 \times K+1$ upper-triangular matrix. Notice that $V = Q_1 R_1$ and $Q_1^T Q_1 = I$. Since rank of $V$ is $K+1$, it follows that $R_1$ is full-ranked and hence invertible. Using elementary matrix algebra it can be shown that $V(V^TV)^{-1}V^T = Q_1 Q_1^T$ and hence $\vec{v}$ satisfies $Q_1Q_1^T \vec{v} = \vec{v}$, which can be rewritten as 
\begin{equation}
\label{EqvMatEq}
(Q_1Q_1^T - I)\vec{v} = 0.
\end{equation}

Let $A = Q_1Q_1^T-I$. Append $v_N$ to $\vec{v}$ and add a last column of all 0's to $A$. Recall that each $v_i \in S$ and we need to determine whether a solution exists for the matrix equation $A\vec{v}=0$ satisfying the Hamiltonian path constraint.

\subsection{Transformation to Exact weight perfect matching problem}
\label{Trans2EWPM}
We now provide equivalences between solving the matrix equation and the exact weight perfect matching problem. Let us currently ignore the Hamiltonian path constraint and also the nature of coefficients of the matrix $A$ and consider solving a more general system $A\vec{v}=0$ where $A$ is some known $M \times N$ matrix and each $v_i \in S$ for some known set $S = \{x_1,x_2,\ldots,x_N\}$. Define a vector $\vec{x} \in \mathbb{Q}^N$ by $\vec{x} = [x_1,x_2,\ldots,x_N]^T$. Checking whether solution exist for the general system $A\vec{v} = 0$ is equivalent to verifying whether there exist a permutation matrix $P$ for which $AP\vec{x} = 0$. Though this problem may have already be shown to be NP-complete, for the sake of completion we provide a proof here by a reduction from set-partition.

\begin{theorem}
\label{APXNPComp}
Given a $M \times N$ matrix $A$ and a vector $\vec{x} \in \mathbb{Q}^N$, determining whether a permutation matrix $P$ exist for which $AP\vec{x} = 0$ is NP-complete.
\end{theorem}
\begin{proof}
This problem can easily be checked to be in $NP$, as given the permutation matrix $P$, verifying whether $AP\vec{x} = 0$ can be done in polynomial time. 

For the NP-hardness proof, consider a special case of the set-partition problem namely, \emph{``Given a set $S$ of $N$ numbers and an integer $T$, does there exist a partition of $S$ into two subsets $G$ and $H$ of cardinality $T$ and $N-T$ respectively, such that the sum of elements in $G$ equals the sum of the elements in $H$?''} Even with this additional restriction on the size of the subset, the set-partition problem is NP-complete as one can add $N$ zeros to a given set-partition problem to get a set of $2N$ elements and check whether this set can be partitioned into two subsets of cardinality $N$ each such that the sum of the elements in one partition equals the sum in the other. We now give a easy reduction from this set-partition problem. 

Define a matrix of size $1 \times N$ where the first $T$ are 1 and the remaining $N-T$ entries are $-1$. As above define $\vec{x} = [x_1,x_2,\ldots,x_N]^T$ where $x_1,x_2,\ldots,x_N$ are elements of the given set $S$. Now if $G = \{g_1,g_2,\ldots,g_T\}$ and $H = \{h_1,h_2,\ldots,h_{N-T}\}$ solves the set-partition problem, then consider the permutation $P$ for which $P\vec{x} = [g_1,g_2,\ldots,g_T,h_1,h_2,\ldots,h_{N-T}]^T$ and hence by construction of $A$, we see that $AP\vec{x} = 0$. The other direction can be proved by just reversing the arguments. 
\end{proof} 
We now reduce the above problem to the exact weight perfect matching problem by defining a weighted complete bipartite graph for each of the $M$ constraints and then later combine these graphs into an equivalently weighted single complete bipartite graph. Imposing the Hamiltonian constraint amounts to enforcing restrictions on the perfect matching which will be done in the end.

Consider a restatement of the above problem. Given a $M \times N$ matrix $A$ and $\vec{x} \in \mathbb{Q}^N$, we need to find whether $\exists \pi \in S_N$ such that $\sum_{j=1}^N a_{i,j}x_{\pi(j)} = 0,1 \leq i \leq M$. Here $a_{i,j}$ denotes the $(i,j)^{th}$ entry of the matrix $A$. Let $m_i = \max \{0,-a_{i,1},-a_{i,2},\ldots,-a_{i,N}\}$ and $n= \max\{0,-x_1,-x_2,\ldots,-x_N\}$. Let $s_i = \sum_{j=1}^N a_{i,j}$ and $s = \sum_{j=1}^N x_j$. Let $b_{i,j} = a_{i,j}+m_i$ and $y_j = x_j+n$. Define $\alpha_i = ns_i+m_is+Nnm_i$. By definition we have $b_{i,j},y_j \geq 0$. The problem is then equivalent to finding a permutation $\pi \in S_N$ such that $\sum_{j=1}^N b_{i,\pi(j)}y_{\pi(j)} = \alpha_i$ for $1 \leq i \leq M$. With the above set up in place, each constraint can be turned into an exact weight perfect matching problem in a complete bipartite graph as follows.

Define a complete bipartite graph $G=(V,E)$, $V = B \cup Y$ where $|B| = |Y| = N$. Index the vertices of $B$ and $Y$ as $\tilde{b_1},\tilde{b_2},\ldots,\tilde{b_N}$ and $\tilde{y_1},\tilde{y_2},\ldots,\tilde{y_N}$ respectively. Let $e_{u,v} \in E$ denote the edge connecting the vertices $\tilde{b_u}$ and $\tilde{y_v}$. For each one of the above $M$ constraints, define a cost $c_i(e_{u,v})$ for the edge $e_{u,v}$ as $c_i(e_{u,v}) = b_{i,u}y_v, 1 \leq i \leq M$. The above problem can be reformulated as \emph{``Given this graph $G$ with costs defined as above, does there exist a perfect matching $PM \subset B \times Y$ such that the cost of the $PM$, $(C_i(PM))$ equals $\alpha_i$ for all $1 \leq i \leq M$?''}

We now transform these $M$ exact weight perfect matching problem into one problem by following the lines of \cite{Deineko06}, where the authors reduce the min-max assignment problem \cite{Aissi05} to exact weight perfect matching problem \cite{Deineko06}. Let $C_{max} = \max_{i,u,v} c_i(e_{u,v})$. Consider a complete bipartite graph $\tilde{G}$ with the same vertex set as graph $G$. Define the cost $w(e_{u,v})$ for the edge $e_{u,v}$ in $\tilde{G}$ as 
\begin{equation}
\label{edgeWeight}
w(e_{u,v}) = \sum_{i=1}^M b_{i,u}y_v (NC_{max}+1)^{i-1} = h_u y_v, 
\end{equation}
where
\begin{equation}
\label{summedWeights}
h_u =\sum_{i=1}^M b_{i,u}(NC_{max}+1)^{i-1}.
\end{equation}
Define 
\begin{equation}
\label{alphaDef}
\alpha = \sum_{i=1}^M \alpha_i (NC_{max}+1)^{i-1}. 
\end{equation}
We now show the following.
\begin{theorem}
G has a perfect matching $PM$ with $C_i(PM) = \alpha_i, 1 \leq i \leq M$ if and only if, $\tilde{G}$ has a perfect matching $PM$ with $C(PM) = \alpha$.
\end{theorem}
\begin{proof}
For both directions of the proof, let $e_{1,\pi(1)},e_{2,\pi(2)},\ldots,e_{N,\pi(N)}$ for some $\pi \in S_N$ constitute the edges of the perfect matching.\\\\
case(i): Let $G$ have a perfect matching $PM$ with $C_i(PM) = \alpha_i$. We have $\sum_{j=1}^N c_i(e_{j,\pi(j)}) = \alpha_i, 1 \leq i \leq M$. Then, 
\begin{equation*}
\sum_{j=1}^Nw(e_{j,\pi(j)}) = \sum_{j=1}^N\sum_{i=1}^Mc_i(e_{j,\pi(j)}) (NC_{max}+1)^{i-1} = \sum_{i=1}^M \alpha_i(NC_{max}+1)^{i-1} = \alpha  
\end{equation*}\\\\
case(ii): Let $\tilde{G}$ have a perfect matching $PM$ with $C(PM) = \alpha$, i.e $\sum_{j=1}^N w(e_{j,\pi(j)}) = \alpha$. It then follows,
\begin{equation*}
\alpha = \sum_{j=1}^N\sum_{i=1}^Mc_i(e_{j,\pi(j)}) (NC_{max}+1)^{i-1} = \sum_{i=1}^N (NC_{max}+1)^{i-1}\left[\sum_{j=1}^Nc_i(e_{j,\pi(j)})\right]
\end{equation*}
Recall the definition of $\alpha$ from equation(\ref{alphaDef}). Also notice the inequality 
\begin{equation*}
\sum_{j=1}^N c_i(e_{j,\pi(j)}) \leq NC_{max}, \forall i, 
\end{equation*}
as $C_{max} \geq c_i(e_{u,v}) $ and $c_i(e_{u,v})\geq 0$ by definition. By representing numbers in base $(NC_{max}+1)$, it follows that $\sum_{j=1}^N c_i(e_{j,\pi(j)}) = \alpha_i = C_i(PM), \forall i$ which completes the proof.   
\end{proof}
We would like to emphasize that $\tilde{G}$ is \emph{polynomial time} reducible from $G$ and $\alpha$ can be represented in $O(M \log C_{max})$ bits. Thus the following result can be deduced from the above reduction and theorem \ref{APXNPComp}, namely
\begin{corollary}
\label{ExactMatchNPComp}
Exact weight perfect matching is NP-complete
\end{corollary} 
which is a known result \cite{Aissi05,Deineko06}.

\subsection{Transformation to Multi-Criteria Traveling salesman problem}
\label{Trans2MultiTSP}
Barring the Hamiltonian path constraint, we have shown how the K-compressibility problem can be reduced in polynomial time to the exact weight perfect matching problem by a series of transformations from the equation (\ref{EqvMatEq}), as illustrated in section (\ref{Trans2EWPM}). Enforcing the Hamiltonian path constraint reduces it to the following permutation problem.

Define a set $Z$ consisting of $N$ tuples $Z = \left\{\left(\begin{matrix}h_1\\y_1\end{matrix}\right),\left(\begin{matrix}h_2\\y_2\end{matrix}\right),\ldots,\left(\begin{matrix}h_N\\y_N\end{matrix}\right)\right\}$ where $h_i's$ are defined according to equation(\ref{summedWeights}) and $y_i = x_i + n$ as defined above. Recall that $x_N$ is chosen to be the element for which we have assumed the existence of $x_0 \in S$ and $f \in \prod_K$ such that $f^{N-1}(x_0) = x_N$ with the property that $f^{i}(x_0) \in S, 0 \leq i \leq N-1$. In other words $x_N$ is assumed to be the last element of $S$ in the orbit of $f$ starting from $x_0$.
 
Verifying whether the given set $S$ is K-compressible with the $x_N$ being the last element is equivalent to checking the existence of a permutation $\pi \in S_N$ such that
\begin{equation}
\label{ExactCostPerm}
\sum_{u=1}^N h_{\pi(u)}y_{\pi(u+1)} = \alpha
\end{equation}
where $\alpha$ is defined as per equation(\ref{alphaDef}) and $\pi(N+1)$ refers to $\pi(1)$. This problem of finding the permutation $\pi$ can be transformed to a multi-criteria traveling salesman problem \cite{Blaser08,Manthey09a,Manthey09b}, as illustrated below.

Define a weighted complete directed graph $G' = (V',E')$ consisting of $N$ vertices with no self-loops. Let us label the vertices in $G'$ as $x_1,x_2,\ldots,x_N$ corresponding to the elements of the given set $S$ for which we need to check whether it is K-compressible. Let $e'_{u,v}$ denote the directed edge from vertex $x_u$ to $x_v$. Define the weight on edge $e'_{u,v}$ to be $w(e'_{u,v}) = h_u y_v$. The problem then reduces to determining whether there exist a tour in $G'$ such that each vertex is visited exactly once and the cost of the tour being \emph{exactly} equal to $\alpha$. De$\breve{i}$neko and Woeginger \cite{Deineko06} have shown polynomial reduction of the exact weight perfect matching problem to the min-max assignment problem. The same reduction technique can be used to transform the exact cost tour problem to a multi-criteria traveling salesman problem \cite{Manthey09b}, solutions to which can then be availed to get fast solutions to our problem. To decide whether the given set $S$ is K-compressible, we just need to repeat this process over all the $N$ possible choices for $x_N$.

\section{Computational complexity of K-compressibility}
\label{ComputationalComplexity}
The previous section provided approaches to solve the K-compressibility problem by reducing it to the multi-criteria traveling salesman problem \cite{Manthey09b}. But the above reduction doesn't throw any light on the inherent complexity of our problem. The problem is indeed in $NP$ as the certificate can be the actual sequence in which the elements of $S$ occurs in the orbit. Once the sequence is known, verifying whether a $K^{th}$ polynomial $f$ can produce elements of $S$ in that sequence can easily be done in polynomial time as discussed under section(\ref{introduction}). But the following question still remains unanswered,\emph{''Is K-compressibility NP-hard?''} 

Reconsider the above set of equivalent transformations discussed under section (\ref{EqvTrans}) where we reduced our problem to the multi-criteria traveling salesman problem. The multi-criteria traveling salesman problem is a well known NP-complete problem for an arbitrary set of weights on the edges \cite{Manthey09b}. But if the edge weights are defined according to the transformations described in section (\ref{EqvTrans})--which are some complicated functions (involving $QR$ decomposition of the Vandermonde matrix) of the set elements $x_1,x_2,\ldots,x_N$--is the resultant traveling salesman problem still NP-complete? Currently, we do not have an affirmative answer to this question. But we now show that even for the simplest case where the edge weight $w(e'_{u,v})$ for the edge $e'_{u,v} \in E'$ connecting the vertices $x_u$ and $x_v$ in the resultant graph $G'$--defined in section(\ref{EqvTrans})--is $w(e'_{u,v}) = y_uy_v$ where we have replaced $h_u$ (defined in equation (\ref{summedWeights})) with $y_u$, the multi-criteria traveling salesman problem or equivalently exact cost tour traveling salesman problem is NP-complete.

With the weights defined as above, namely $w(e'_{u,v}) = y_uy_v$, the exact cost tour traveling salesman problem is equivalent to determining whether $\exists \pi \in S_N$ such that $\sum_{i=1}^{N} y_{\pi(i)} y_{\pi(i+1)} = \alpha$ for some fixed cost $\alpha$ with the understanding that $\pi(N+1) = \pi(1)$. We now show the problem to be NP-complete by a reduction from bounded-knapsack problem.

\begin{theorem}
Given a set $S = {y_1,y_2,\ldots,y_N}$ of $N$ elements and a cost $\alpha$, determining whether $\exists \pi \in S_N$ such that $\sum_{i=1}^{N} y_{\pi(i)} y_{\pi(i+1)} = \alpha$ is NP-complete.
\end{theorem}
\begin{proof}
The problem can easily be seen to be in $NP$, as once the permutation $\pi$ is known, verification can be done in polynomial time. For the NP-hardness proof we consider reduction from bounded-knapsack problem. 

Bounded-knapsack problem is defined as follows. \emph{``Given a set of items, each with utmost $b_i$ copies, and a weight and profit for each of it, determine the number of each item to include in a collection so that the total weight is less than a given weight limit and the total profit exceeds the given minimum profit?''} Let $W = \{w_1,w_2,\ldots,w_N\}$ and $P = \{p_1,p_2,\ldots,p_N\}$ denote the weights and profits respectively. Let $\alpha_W$ and $\alpha_P$ be the weight limit and profit value respectively. The bounded-knapsack problem is to determine whether the following inequalities
\begin{equation*}
\sum_{i=1}^N p_i \gamma_i \geq \alpha_P, \hspace{0.2in} \sum_{i=1}^N w_i \gamma_i \leq \alpha_W
\end{equation*}
can be satisfied with $\gamma_i \in \{0,1,\ldots,b_i\}$. 

Let us consider an special case of the bounded-knapsack problem where $p_i = w_i = y_i$ and $b_i = 2, \forall i$. Having the weights and profits to be equal, reduces it to the subset-sum problem with two copies for each item. The problem is then meaningful only when $\alpha_P = \alpha_W = \alpha$. It reduces to checking whether $\exists \gamma_i \in \{0,1,2\}$ such that $\sum_{i=1}^N y_i \gamma_i = \alpha$. We now show reduction from this NP-complete problem.

Without loss of generality assume that $y_i, \alpha \in \mathbb{Z^+}$ and $y_i > 1$. Define $z_i = \alpha y_i$. The problem is equivalent to verifying whether $\sum_{i=1}^N z_i \gamma_i = \alpha^2 $ with $\gamma_i \in \{0,1,2\}$. Let $S' = \{z_1,z_2,\ldots,z_N\}$. Let $ p = 2N+1$. Append the set $S'$ with $\frac{1}{p}$ $N+1$ times and $2N+2$ 0's and call the appended set $S''$. This set $S''$ and the cost $\frac{\alpha^2}{p}$ will be the input to our problem.\\\\
case(i): Assume that the bounded-knapsack problem has a solution ,i.e, $\exists \gamma_i \in \{0,1,2\}$ such that $\sum_{i=1}^N z_i \gamma_i = \alpha^2$. Let $T_0$, $T_1$ and $T_2$ form a partition of $S'$ such that $z_i \in T_n$ if and only if $\gamma_i = n$ in the solution. Let $N_n$ denote the cardinality of the subset $T_n$ whose elements are denoted by $T_n = \{t_{n,1}, t_{n,2},\ldots,t_{n,N_n}\}$. Consider the permutation $\pi$ which places the elements of $S''$ in the order 
\begin{eqnarray*}
&&t_{0,1},0,t_{0,2},0,\ldots,t_{0,N_0},0,t_{1,1},\frac{1}{p},t_{1,2},0,t_{1,3},\frac{1}{p},t_{1,4},\ldots,0,t_{1,N_1-1},\frac{1}{p},t_{1,N_1},0,\\
&&\frac{1}{p},t_{2,1},\frac{1}{p},t_{2,2},\frac{1}{p},t_{2,3},\ldots,\frac{1}{p},t_{2,N_2},\frac{1}{p},0,\frac{1}{p},0,\ldots,0
\end{eqnarray*}
i.e, every $t_{0,i} \in T_0$ is sandwiched between two $0's$, each $t_{1,i} \in T_1$ is \emph{either} succeeded or preceded by $\frac{1}{p}$ and every $t_{2,i} \in T_2$ is sandwiched between two $\frac{1}{p}$. By construction, this permutation $\pi$ is the solution to our problem.\\\\
case(ii): Let the permutation $\pi$ be the solution to our problem.  Since $z_i > \alpha$, no two elements of $S'$ can be juxtaposed in the ordering induced by $\pi$. Let
\[
\gamma_i = \left\{
\begin{array}{ll}
0&\mbox{if $z_i$ is sandwiched between two 0's in the order};\\
1&\mbox{if $z_i$ is \emph{either} preceded or succeded by $\frac{1}{p}$};\\
2&\mbox{if $z_i$ is sandwiched between two $\frac{1}{p}$ in the order}.
\end{array}
\right.
\]
We have,
\begin{equation*}
\frac{\sum_{i=1}^N z_i \gamma_i}{p} + \frac{\beta}{p^2} = \frac{\alpha^2}{p}
\end{equation*}
where $0 \leq \beta \leq 2N$ and $\frac{\beta}{p^2}$ is the sum obtained from juxtaposing $\frac{1}{p}$. Rearranging terms we get $\beta = (\alpha^2 - \sum_{i=1}^N z_i \gamma_i) p$. Since $\alpha, z_i \in \mathbb{Z}$ and $p=2N+1$, $\beta$ should equal $0$.
Thus $\sum_{i=1}^N z_i \gamma_i = \alpha^2$ which completes the proof.
\end{proof}
From the above theorem it easily follows that, for the general case where $h_u$ is arbitrary--instead of equation (\ref{summedWeights})--the exact cost tour traveling salesman problem with weights defined as above is NP-complete. This provides a strong evidence that the K-compressibility problem may not have a polynomial time solution unless P=NP.

\subsection{Difficulty in NP-completeness proof}
An answer to the following question might throw some light on the difficulty of showing the K-compressibility problem to be NP-hard. \emph{``For every $K \geq 2$, does $\exists$ an integer $M(K)$, such that $\forall N \geq M(K)$, any K-compressible set $S$ consisting of $N$ distinct rationals is uniquely K-compressible, i.e there exist utmost one compressing function $f \in \prod_K, f \notin \prod_1$?''}. In other words, is it true that, for a given $K \geq 2$, for a sufficiently large N, any set $S = \{x_1,\cdots,x_N\}$ with $x_i \in \mathbb{Q}$ that is K-compressible and not 1-compressible, is \emph{uniquely} K-compressible? It is worth emphasizing that the result is not true for $K=1$. Since the inverse of a linear function is linear, any set $S$ compressible by a linear function $f(x) = ax+b$ is also compressible by its inverse $g(x) = f^{-1}(x) = \frac{x-b}{a}$, with $g$ producing the elements in the reverse order as given by $f$. The question is definitely interesting for $K \geq 2$ as it may provide insight on the uniqueness of the orbits of polynomial functions. 

Our hunch is that, the aforementioned claim, namely the existence of the bound $M(K)$ for all $K \geq 2$, is indeed true. The verity of our claim has farfetched implications on the K-compressible problem being NP-complete. At this stage, it is worth recalling the definition of the complexity class \emph{UP} which stands for \emph{``Unambiguous Non-deterministic Polynomial-time''}. The definition of this class is as follows: A language $L$ belongs to $UP$ if there exists a two input polynomial time algorithm $\mathbb{A}$ and a constant $\alpha$ such that
\begin{itemize}
\item if x in L , then there exists a \emph{unique} certificate $y$ with $|y| = O(|x|^{\alpha})$ such that $\mathbb{A}(x,y) = 1$.
\item if x isn't in L, there is \emph{no} certificate $y$ with $|y| = O(|x|^{\alpha})$ such that $\mathbb{A}(x,y) = 1$.
\end{itemize}
The algorithm $\mathbb{A}$ verifies $L$ in polynomial time. The crucial aspect of the class $UP$ is the uniqueness of the certificate $y$ if one exists. If the bound $M(K)$ exists, then for all sets $S$ with cardinality $N \geq M(K)$, either the set is not K-compressible or it is \emph{uniquely} K-compressible. Then the K-compressible problem with $N \geq M(K)$ may actually belong to class $UP$. The question of whether $UP =  NP$, is still an open problem in the theoretical computer science community. 

For the trivial case where $K=1$, we would like to prove the following result.

\begin{theorem}
Any 1-compressible set $S$, with $|S| \geq 5$, is uniquely compressible modulo the function inverse.
\end{theorem}

\begin{proof}
Let $S = {x_1,x_2,\cdots,x_N}$ be the given set and without loss of generality we can assume that $x_i < x_{i+1}$. Then the only possible ways in which the set $S$ can be produced by iterating a linear function $f=ax+b$ are \\ \\
(i) $x_1,x_2,\cdots,x_N$ and its reverse order $x_N,x_{N-1},\cdots,x_1$, when $a>0$\\
or\\
(ii) $x_1,x_N,x_2,x_{N-1},\cdots$ and its reverse order $\cdots,x_{N-1},x_2,x_N,x_1$, when $a<0$\\
or\\
(iii) $x_N,x_1,x_{N-1},x_2,\cdots$ and its reverse order $\cdots,x_2,x_{N-1},x_1,x_N$, when $a<0$.\\ \\
For any set produced in the order (i), it is easy to see that $\frac{x_{i+1}-x_{i}}{x_{i}-x_{i-1}} = a$ (by the definition of the slope of a line), i.e, the difference between the consecutive elements either increases (for $a>1$) or decreases (for $a<1$) or remains constant (for $a=1$). But for any set produced by in the order (ii) or (iii), the difference between the consecutive elements intially decreases and then steadily increases, specifically we have, $x_3 - x_2 < x_2-x_1$ and $x_{N-1}-x_{N-2} < x_N - x_{N-1}$. This can be seen by defining $h \equiv f^2 \equiv a^2x+c$ (the composition of $f$ with itself), for some constant $c$. Then we have, $h(x_1) = x_2, h(x_2) = x_3,\cdots, h(x_N)=x_{N-1},h(x_{N-1})=x_{N-2}$. The series $x_1,x_2,\cdots$ and the series $x_N,x_{N-1},\cdots$ converges to the common point $x_0$ which is the only fixed point of $h$, i.e, $h(x_0)=x_0$. Hence the difference between the consecutive elements of $S$ (represented in the ascending order) will steadily decrease for a while and then steadily increase. Hence it follows that for sufficiently large $N$, any set produced in the order (i) cannot be produced in the order (ii) or (iii) and vice versa.

Now, we just need to show that for any set produced in the order (ii), it cannot be produced in the order (iii). To this end, let $f_1(x)=a_1x+b_1$, produce the set $S$ in the order (ii), namely $x_1,x_N,x_2,x_{N-1},\cdots$, with $f_1(x_1)=x_N,f_1(x_N)=x_2,\cdots$. Similarly, let $f_2(x)=a_2x+b_2$ produce the same set $S$ in the order (iii), namely $x_N,x_1,x_{N-1},x_2,\cdots$, with $f_2(x_N)=x_1,f_2(x_1)=x_{N-1},\cdots$. Defining $h(x) \equiv f_2 \circ f_1(x) \equiv a_2a_1x+a_2b_1+b_2$, we have $h(x_1)=x_1,h(x_2)=x_2,\cdots,h(x_N)=x_N$. Then $h(x)=x, \forall x$, as it is the \emph{only} linear function which can have  more than one fixed point. Then $f_2 = f_1^{-1}$. This results in a contradiction as $f_1^{-1}(x_2)=x_N$ but $f_2(x_2)=x_{N-2}$ by definition.

It is easy to check that for $N \geq 5$, the aforementioned arguments holds.
\end{proof} 

It is an interesting open question to prove or disprove the existence of the such bound, $M(K)$, for $K \geq 2$. We strongly believe $M(K)$ does exist and furthermore $M(K) = \Theta(K)$, i.e, varies \emph{linearly} with $K$. We also observed in our experiements that, given a set $S$ of size $N$, the number of possible permutations by which it is compressible using a $K^{th}$ degree polynomial drops down exponentially from $N!$ for $K=N-2$, to one when $K = O(N)$. Such exponential drop is uncharacteristic of NP-complete problems.

\subsection{Experimental results}
We now corroborate our claim with the following experimental results. Firstly, we observed that the set $S = \{1,2,\cdots,N-1,N\}$ is one of the very few sets which are compressible for many different permutations of $S$ for a given $K$. Since we did a brute force search over all the $N!$ possible permutations of $S$, we could conduct the experiment only for small values of $N$ ranging from 6 to 14. Since a $K^{th}$ degree polynomial subsumes the $K-1^{th}$ degree polynomial and the linear functions $f_1(x)=x+1$ and $f_2(x)=x-1$ can compress $S$, we can stop at the value of $K$ where the number of possible permutations equals 2. The results adumbrated in the table below unveil that this stopping value of $K$ is not far away from $N-2$.

\begin{center}
\begin{table}[ht!]
\begin{minipage}[c]{0.3\linewidth}\centering
\caption{$N=14$}
\begin{tabular}{|l|l|}
\hline 
K  & \# permutations\tabularnewline
\hline 
12  & 14!\tabularnewline
\hline 
11  & 836644\tabularnewline
\hline 
10  & 24\tabularnewline
\hline
9  & 4\tabularnewline
\hline
8  & 2\tabularnewline
\hline
\end{tabular}
\end{minipage}
\hspace{0.1cm}
\begin{minipage}[c]{0.3\linewidth}\centering
\caption{$N=13$}
\begin{tabular}{|l|l|}
\hline 
K  & \# permutations\tabularnewline
\hline 
11  & 13!\tabularnewline
\hline 
10  & 394204\tabularnewline
\hline 
9  & 46\tabularnewline
\hline
8  & 2\tabularnewline
\hline
7  & 2\tabularnewline
\hline
\end{tabular}

\end{minipage}
\hspace{0.1cm}
\begin{minipage}[c]{0.3\linewidth}\centering
\caption{$N=12$}
\begin{tabular}{|l|l|}
\hline 
K  & \# permutations\tabularnewline
\hline 
10  & 12!=479001600\tabularnewline
\hline 
9  & 23070\tabularnewline
\hline 
8  & 28\tabularnewline
\hline
7  & 4\tabularnewline
\hline
6  & 2\tabularnewline
\hline
\end{tabular}
\end{minipage}
\end{table}
\par\end{center}

\begin{center}
\begin{table}[ht!]
\begin{minipage}[c]{0.3\linewidth}\centering
\caption{$N=11$}
\begin{tabular}{|l|l|}
\hline 
K  & \# permutations\tabularnewline
\hline 
9  & 11!=39916800\tabularnewline
\hline 
8  & 11798\tabularnewline
\hline 
7  & 32\tabularnewline
\hline
6  & 2\tabularnewline
\hline
5  & 2\tabularnewline
\hline
\end{tabular}
\end{minipage}
\hspace{0.1cm}
\begin{minipage}[c]{0.3\linewidth}\centering
\caption{$N=10$}
\begin{tabular}{|l|l|}
\hline 
K  & \# permutations\tabularnewline
\hline 
8  & 10!=3628800\tabularnewline
\hline 
7  & 1726\tabularnewline
\hline 
6  & 12\tabularnewline
\hline
5  & 4\tabularnewline
\hline
4  & 2\tabularnewline
\hline
\end{tabular}
\end{minipage}
\hspace{0.1cm}
\begin{minipage}[c]{0.3\linewidth}\centering
\caption{$N=9$}
\begin{tabular}{|l|l|}
\hline 
K  & \# permutations\tabularnewline
\hline 
7  & 9!=362880\tabularnewline
\hline 
6  & 398\tabularnewline
\hline 
5  & 8\tabularnewline
\hline
4  & 2\tabularnewline
\hline
3  & 2\tabularnewline
\hline
\end{tabular}
\end{minipage}
\end{table}
\par\end{center}

\begin{center}
\begin{table}[ht!]
\begin{minipage}[c]{0.3\linewidth}\centering
\caption{$N=8$}
\begin{tabular}{|l|l|}
\hline 
K  & \# permutations\tabularnewline
\hline 
6  & 8!=40320\tabularnewline
\hline 
5  & 68\tabularnewline
\hline 
4  & 2\tabularnewline
\hline
\end{tabular}
\end{minipage}
\hspace{0.1cm}
\begin{minipage}[c]{0.3\linewidth}\centering
\caption{$N=7$}
\begin{tabular}{|l|l|}
\hline 
K  & \# permutations\tabularnewline
\hline 
5  & 7!=5040\tabularnewline
\hline 
4  & 66\tabularnewline
\hline 
3  & 2\tabularnewline
\hline
\end{tabular}
\end{minipage}
\hspace{0.1cm}
\begin{minipage}[c]{0.3\linewidth}\centering
\caption{$N=6$}
\begin{tabular}{|l|l|}
\hline 
K  & \# permutations\tabularnewline
\hline 
4  & 6!=720\tabularnewline
\hline 
3  & 8\tabularnewline
\hline 
2  & 2\tabularnewline
\hline
\end{tabular}
\end{minipage}
\end{table}
\par\end{center}

From the table, it is clear that the number of permutations by which the given set $S$ (compressing of integers from 1 to $N$) is compressible, using a $K^{th}$ degree polynomial, drops down exponentially faster from $N!$ for $K=N-2$ to just 2 when $K \leq N-6$.

Next, for small values of $K$ ($K=2,3$), we did a brute force search over all the possible sets of a given cardinality $N$, consisting of integers from a broad range $[-R,R]$, to verify whether the bound $M(K)$ exist. Since we restrict ourselves to sets $S$ over the rationals $\mathbb{Q}$ (only for which the notion of complexity is well defined), from theorem (\ref{affineTrans}) it follows that, considering sets over $\mathbb{Q}$ is equivalent to considering sets over the integers $\mathbb{Z}$. Without loss of generality, we can also assume that $0 \in S$ (as the elements of the set can be moved so that one of the element is 0) and it is the starting element from which the other numbers of the set are produced by iterating a polynomial function. We now illustrate in detail for the case $K=2$. The test for $K=3$ was done in the similar fashion.

We considered \emph{all possible triplets} of numbers $x_1, x_2, x_3$ satisfying the following set of constraints, namely
\begin{itemize}
\item $x_i \not= x_j$ if $i \not= j$.
\item $x_i \not= 0$, $i \in \{1,2,3\}$.
\item $-R \leq x_i \leq R$,$i \in \{1,2,3\}$. 
\item The sequence $0, x_1, x_2, x_3$ cannot be produced by iterating a linear function.
\end{itemize}
We set $R=10,000$. We then constructed a quadratic function $f=ax^2+bx+c$ satisfying $f(0)=x_1$, $f(x_1)=x_2$ and $f(x_2)=x_3$. Upon iterating $f$, we obtained the remaining numbers of the set $S$. Since the numbers of $S$ can be as large as $10^{40}$, we used the GMP package--acronymed for GNU Multiple Precision library--which provides arbitrary precision arithmetic. When the size of $S$ equals 6,7 or 8, we found no other permutation of its numbers for which the generated set can be produced by iterating a different quadratic function. Hence we strongly believe that the bound $K=2$ is $M(K)=6$. Similarly for $K=3$, we conjecture that $M(K)=8$.

The experiments seems to suggest the the bound $M(K)$ may exist even for other values of $K$. If it does exists, then the K-compressibility problem will be in the complexity class $UP$ whenever $N \geq M(K)$. Hence we believe that, unless $UP=NP$, it may be hard to show that the K-compressibility problem is infact NP-hard.
\section{$\epsilon$ K-Compressibility of sets}
\label{epsilonCompressibility}
The previous sections dealt with lossless compression where the elements of the given set $S$ are reproduced \emph{exactly} by iterating a polynomial $f$ of a relatively smaller degree $K$. Exact representation of a set is typically a non-requirement and most real world scenarios do not impose such a hard constraint. In this section, we formally define the notion of $\epsilon$ K-compressibility and then provide the \emph{necessary} condition that a given set $S$ should satisfy for it to be K-compressible while allowing for small perturbation of its elements up to an $\epsilon$ distance.

We start with defining the notion of $\epsilon$-perturbation.
\begin{definition}
Given $\epsilon > 0$, $y$ is an $\epsilon$-perturbation of $x$ if $|y-x| \leq \epsilon |x|$.
\end{definition}

\begin{definition}
A set $S = \{x_1,x_2,\ldots,x_N\}$ consisting of $N$ numbers is $\epsilon$ K-compressible if there exist a K-compressible set $\tilde{S} = \{\tilde{x}_1,\tilde{x}_2,\ldots,\tilde{x}_N\}$ where $\tilde{x}_i$ is the $\epsilon$-perturbation of $x_i$.
\end{definition}
Intuitively, it is clear that larger the value of $\epsilon$, better is the compression ratio. But given an $\epsilon$ and $K< N-2$, how to determine whether the set is $\epsilon$ K-compressible? Closely following the lines discussed under section (\ref{EqvTrans}), we now provide an approach to tackle this problem.

As the elements of the given set $S$ can be shifted and the given $\epsilon$ suitably adjusted, we can restrict ourselves to the case where the elements of the set $S$ are assumed to be positive, i.e, $x_i > 0, \forall i$. Let $\tilde{x}_i$ denote the $\epsilon$-perturbation of $x_i$. Pick a $y \in S$ and assume its $\epsilon$-perturbation to be the last element in the orbit of a $K^{th}$ degree polynomial $f$ starting from $\tilde{x}_0$, i.e, $f^{N-1}(\tilde{x}_0) = \tilde{x}_N$, where $\tilde{x}_0$ is an $\epsilon$-perturbation of some element $x_0 \in S$. Without loss of generality let $x_N = y$. Let $x_i < x_{i+1} ,1 \leq i \leq N-2$. Assume $\epsilon$ to be small enough so that $(1+\epsilon)x_i < (1-\epsilon)x_{i+1}$ ,i.e, $\tilde{x}_i < \tilde{x}_{i+1}$. The problem then amounts to finding the $\tilde{x}_i's$ such that the following over-determined Vandermonde system
\[ 
\left[
\begin{matrix}
1&\tilde{x}_1&\tilde{x}_1^2&\dots&\tilde{x}_1^K\\
1&\tilde{x}_2&\tilde{x}_2^2&\dots&\tilde{x}_2^K\\
\vdots&\vdots&\vdots&\vdots&\vdots\\
1&\tilde{x}_{N-1}&\tilde{x}_{N-1}^2&\dots&\tilde{x}_{N-1}^K
\end{matrix}
\right]
\left[
\begin{matrix}
a_0\\a_1\\\vdots\\a_K
\end{matrix}
\right] = \left[
\begin{matrix}
v_1\\v_2\\\vdots\\v_{N-1}
\end{matrix}
\right]
\]
has a solution which also satisfies the Hamiltonian path constraint. Each $v_i$ is an $\epsilon$ perturbation of some $x_i \in S$ and equal to one of the $\tilde{x}_j's$. $v_N$ is the starting element $\tilde{x}_0$. Recall that in section(\ref{EqvTrans}) we proceeded forward by considering the $QR$ decomposition of the above Vandermonde matrix.  Since $\tilde{x}_i's$ are by themselves unknown and are $\epsilon$ perturbations of $x_i$, bounding the coefficients of the matrices $Q$ and $R$ as a function of $\epsilon$ seems to be hard problem by itself. Hence we consider a slightly different formulation which involves computing determinants of Vandermonde matrices for which closed form solutions exist and hence the coefficients of the resultant matrix can easily be bounded.

Finding the $\tilde{x}_i's$ for which the above Vandermonde system has a solution is equivalent to enforcing that the determinant of \emph{every} $K+2 \times K+2$ sub-matrix of the following matrix
\[
\left[
\begin{matrix}
1&\tilde{x}_1&\tilde{x}_1^2&\dots&\tilde{x}_1^K&v_1\\
1&\tilde{x}_2&\tilde{x}_2^2&\dots&\tilde{x}_2^K&v_2\\
\vdots&\vdots&\vdots&\vdots&\vdots&\vdots\\
1&\tilde{x}_{N-1}&\tilde{x}_{N-1}^2&\dots&\tilde{x}_{N-1}^K&v_{N-1}
\end{matrix}
\right]
\]
is zero. Keeping the first $K+1$ rows fixed and changing the $K+2$ row and expanding w.r.t to the last column, we obtain the following $N-K-2$ linear equations on $v_1,v_2,\ldots,v_{N-1}$ namely,
\begin{eqnarray}
\label{EqvLinearEq}
\tilde{a}_{1,1}v_1+\tilde{a}_{1,2}v_2+\cdots+\tilde{c}v_{K+2}+0v_{K+3}+\cdots+0v_{N-1} &=& 0 \nonumber \\
\tilde{a}_{2,1}v_1+\tilde{a}_{2,2}v_2+\cdots+0v_{K+2}+\tilde{c}v_{K+3}+\cdots+0v_{N-1} &=& 0 \nonumber \\
\vdots\nonumber \\
\tilde{a}_{N-K-2,1}v_1+\tilde{a}_{N-K-2,2}v_2+\cdots+0v_{K+2}+\cdots+\tilde{c}v_{N-1}&=& 0
\end{eqnarray}
where $\tilde{a}_{i,j} $ and $c$ are the determinants of the following Vandermonde matrices respectively. 
\[ \tilde{a}_{i,j} = (-1)^{K+j}\left|
\begin{matrix}
1&\tilde{x}_1&\tilde{x}_1^2&\dots&\tilde{x}_1^K\\
\vdots&\vdots&\vdots&\vdots&\vdots\\
1&\tilde{x}_{j-1}&\tilde{x}_{j-1}^2&\dots&\tilde{x}_{j-1}^K\\
1&\tilde{x}_{j+1}&\tilde{x}_{j+1}^2&\dots&\tilde{x}_{j+1}^K\\
\vdots&\vdots&\vdots&\vdots&\vdots\\
1&\tilde{x}_{K+1}&\tilde{x}_{K+1}^2&\dots&\tilde{x}_{K+1}^K\\
1&\tilde{x}_{K+1+i}&\tilde{x}_{K+1+i}^2&\dots&\tilde{x}_{K+1+i}^K
\end{matrix}
\right|
\]
\[ \tilde{c} = \left|
\begin{matrix}
1&\tilde{x}_1&\tilde{x}_1^2&\dots&\tilde{x}_1^K\\
1&\tilde{x}_2&\tilde{x}_2^2&\dots&\tilde{x}_2^K\\
\vdots&\vdots&\vdots&\vdots&\vdots\\
1&\tilde{x}_{K+1}&\tilde{x}_{K+1}^2&\dots&\tilde{x}_{K+1}^K
\end{matrix}
\right|
\]
For succinct representation, we can add an additional insignificant term $0*v_N$ to the left hand side of each these equations. The equations can then be represented in the matrix form $\tilde{A}*\vec{v} = 0$ where $\tilde{A}$ is the $N-K-2 \times N$ matrix containing the coefficients (with the last column of all 0's) and $\vec{v} = [v_1,v_2,\ldots,v_{N}]^T \in \mathbb{Q}^N$ representing the variables. If we denote $\vec{\tilde{x}} = [\tilde{x}_1, \tilde{x}_2,\ldots,\tilde{x}_N]^T \in \mathbb{Q}^N$, the problem then reduces to finding a permutation $P$ satisfying the Hamiltonian path constraint such that
\begin{equation}
\label{APVProblem_2}
\tilde{A}*P*\vec{\tilde{x}} = 0.
\end{equation}
Define a $N-K-2 \times N$ matrix $A$ in the exact way we defined $\tilde{A}$ above except replacing the $\tilde{x}_i's$ in the definition with the corresponding $x_i's$. Notice that $A$ can be pre-computed using which we can bound the entries of $\tilde{A}$ as follows. 

To start with consider the determinant of an $K+1 \times K+1$ Vandermonde matrix $V$ which can be written in closed form as
\begin{equation*}
d(\vec{x}) =  d(x_1,x_2,\ldots,x_K) = det(V) = \prod_{1\leq i < j \leq K}(x_j - x_i).
\end{equation*}
Let $\eta_i = \epsilon |x_i|$. For small $\epsilon$, we can use the second order Taylor expansion to approximate
\begin{equation}
\label{quadProg}
d(\vec{x}+\vec{\xi}) \approx d(\vec{x})+\vec{\xi}^T \frac{\partial d}{\partial \vec{x}}+ \frac{1}{2}\vec{\xi}^T \frac{\partial^2 d}{\partial \vec{x}^2} \vec{\xi}
\end{equation}
where the first partials are given by
\begin{equation*}
\frac{\partial d}{\partial x_i} = d(\vec{x})\left(\sum_{k \not=i}\frac{1}{x_i-x_k} \right)
\end{equation*}
and the entries of the Hessian matrix $\frac{\partial^2 d}{\partial \vec{x}^2}$ are given by
\[
\frac{\partial^2 d}{\partial x_i \partial x_j}= \left\{
\begin{array}{ll}
d(\vec{x})\left[\left(\sum_{i \not= k}\sum_{l \not=j}\frac{1}{(x_i - x_k)(x_j - x_l)}\right) + \frac{1}{(x_j - x_i)^2}\right] & \mbox{if $i \not = j$}\\
&\\
d(\vec{x})\left[\left(\sum_{i \not= k}\sum_{l \not=j}\frac{1}{(x_i - x_k)(x_j - x_l)}\right) - \sum_{i \not=k}\frac{1}{(x_i-x_k)^2}\right]& \mbox{if $i=j$}
\end{array}
\right.
\]
Here $\vec{\xi} = [\xi_1,\xi_2,\ldots,\xi_K]^T$ where each component $\xi_i$ is bounded by $\xi_i \leq \eta_i$. Let $d_{max}$ and $d_{min}$ respectively denote the maxima and minima of $d(\vec{x}+\vec{\xi})$ subject to the condition that $\xi_i \leq \eta_i$. Without loss of generality we can assume $x_i < x_{i+1}$ and for small $\epsilon$, $x_i+\eta_i < x_{i+1}- \eta_{i+1}$ and hence both $d(\vec{x}), d(\vec{x}+\vec{\xi}) > 0$. Let $ d_{min} = d(\vec{x})(1-\delta_1)$ and $d_{max} = d(\vec{x})(1+\delta_2)$ where $\delta_1 ,\delta_2 > 0$. Since there is no guarantee for the Hessian matrix to be positive semi-definite, finding the maxima and the minima and hence the bounding coefficients $\delta_1$ and $\delta_2$, of the quadratic program in equation (\ref{quadProg}) may be a NP-hard problem \cite{Pardalos91}. Nevertheless, quadratic programming is well-studied field and the literature is inundated with algorithms which can efficiently compute solutions for a quadratic minimization/maximization problems \cite{Cottle68, Ye92, Fu98, QPpage}. As we are only interested in bouding the value of the determinant, it is not necessary to \emph{exactly} compute the maxima and minima of the quadratic program and a good approximation for the same will suffice.

The above set up can now be employed to bound the entries of $\tilde{A}$ using the entries of $A$. Recall that each entry of both the matrices are either zero or obtained as the determinant of a particular Vandermonde matrix. Running the quadratic program (equation(\ref{quadProg})) for each non-zero entry of $\tilde{A}$, we can easily bound each $\tilde{a}_{i,j}$--$(i,j)^{th}$ entry of $\tilde{A}$--by the corresponding $(i,j)^{th}$ entry of $A$, i.e,
\begin{equation}
\label{aijBound}
(1-\gamma_{1,i,j})|a_{i,j}| \leq |\tilde{a}_{i,j}| \leq (1+\gamma_{2,i,j}) |a_{i,j}|
\end{equation}
where $\gamma_{1,i,j},\gamma_{2,i,j} > 0$.
 Since the entries of both $\tilde{A}$ and $A$ can be signed numbers, we can choose a large enough positive constant $D$ using the above bounds and shift the entries of both the matrices by $D$ so that they are strictly greater than zero. Denote the shifted $A$ and $\tilde{A}$ by matrices $B$ and $\tilde{B}$ respectively and its corresponding shifted $(i,j)^{th}$ entry by $b_{i,j}$ and $\tilde{b}_{i,j}$. Using the bounds for $\tilde{a}_{i,j}$ (equation (\ref{aijBound})) we can bound each $\tilde{b}_{i,j}$ by
\begin{equation}
\label{bijBound}
(1-\delta_{1,i,j})b_{i,j} \leq \tilde{b}_{i,j} \leq (1+\delta_{2,i,j})b_{i,j}
\end{equation}
where
\begin{equation}
\label{delta1Def}
\delta_{1,i,j}= \left\{
\begin{array}{ll}
\frac{\gamma_{1,i,j} a_{i,j}}{b_{i,j}}&\mbox{if $a_{i,j} \geq 0$}\\
&\\
\frac{-\gamma_{2,i,j} a_{i,j}}{b_{i,j}} &\mbox{if $a_{i,j} < 0$}
\end{array}
\right.
\end{equation}
 and
\begin{equation}
\label{delta2Def}
\delta_{2,i,j}= \left\{
\begin{array}{ll}
\frac{\gamma_{2,i,j} a_{i,j}}{b_{i,j}}&\mbox{if $a_{i,j} \geq 0$}\\
&\\
\frac{-\gamma_{1,i,j} a_{i,j}}{b_{i,j}} &\mbox{if $a_{i,j} < 0$}
\end{array}
\right.
\end{equation}
Now if we use the shifted matrix $\tilde{B}$ instead of $\tilde{A}$, the problem of finding the permutation $P$ that solves equation(\ref{APVProblem_2}) is tantamount to finding the permutation $P$ such that
\begin{equation*}
\tilde{B}*P*\vec{\tilde{x}} = \vec{\tilde{\alpha}}
\end{equation*}
where each component $\tilde{\alpha_i}$ of $\vec{\tilde{\alpha}}$ equals $\tilde{\alpha_i} = D*\left(\sum_{j=1}^{N}\tilde{b}_{i,j}\right)$. Again we pick a large enough $C_{\max}$ similar to the one under section(\ref{Trans2EWPM}) and define $\tilde{h}_u$ and $\tilde{\alpha}$ according to equations (\ref{summedWeights}) and (\ref{alphaDef}) respectively, where we replace $b_{i,u}$ and $\alpha_i$ with $\tilde{b}_{i,u}$ and $\tilde{\alpha}_i$. Employing the series of transformation discussed under section(\ref{EqvTrans}), the problem again reduces to finding the permutation $\pi \in S_N$ such that 
\begin{equation}
\label{permProblem}
\sum_{u=1}^N \tilde{h}_{\pi(u)} \tilde{x}_{\pi(u+1)} = \tilde{\alpha}
\end{equation}
where $\pi(N+1)$ refers to $\pi(1)$.
Define $\alpha_i = D*\left(\sum_{j=1}^N b_{i,j}\right)$ where $b_{i,j}$ denote the $(i,j)^{th}$ entry of the matrix $B$. Define $h_u$ and $\alpha$ according to equations (\ref{summedWeights}) and (\ref{alphaDef}) and consider the problem of finding $\pi \in S_N$ such that
\begin{equation*}
\sum_{u=1}^N h_{\pi(u)} x_{\pi(u+1)} = \alpha.
\end{equation*}
The bounds for $\tilde{b}_{i,j}$ using $b_{i,j}$ (equation(\ref{bijBound})) can be used to bound $\tilde{h}_u$ in terms of $h_u$ as
\begin{equation*}
(1-\beta_{1,u})h_u \leq \tilde{h}_u \leq (1+\beta_{2,u})h_u
\end{equation*}
where
\begin{equation*}
\beta_{k,u} = \frac{\sum_{i=1}^{N-K-2} \delta_{k,i,u}b_{i,u} (NC_{\max}+1)^{i-1}}{h_u}
\end{equation*}
for $k \in \{1,2\}$ and $\delta_{1,1,j}$ and $\delta_{2,1,j}$ are given by equations (\ref{delta1Def}) and (\ref{delta2Def}) respectively. It is worth mentioning that $\beta_{k,u} \geq 0$ and since $h_u$ and $\tilde{h}_u$ are both positive, $\beta_{1,u} < 1$. Let $\beta_{k,\max} = \max_u \{\beta_{k,u}\}$ for $k \in \{1,2\}$. Then it is easy to see that $(1-\beta_{1,max})h_u \leq \tilde{h}_u \leq (1+\beta_{2,max})h_u$. Also recollect that $\tilde{x}_u$ is the $\epsilon$-perturbation of $x_u$ and hence by definition $(1-\epsilon)x_u \leq \tilde{x}_u \leq (1+\epsilon)x_u$. Using the bounds for $\tilde{h}_u$ and $\tilde{x}_u$, for any permutation $\pi \in S_N$, we can bound the sum $\sum_{u=1}^N \tilde{h}_{\pi(u)} \tilde{x}_{\pi(u+1)}$ by
\begin{equation}
\label{boundForSum}
q_1\sum_{u=1}^N h_{\pi(u)} x_{\pi(u+1)} \leq \sum_{u=1}^N \tilde{h}_{\pi(u)} \tilde{x}_{\pi(u+1)} \leq q_2\sum_{u=1}^N h_{\pi(u)} x_{\pi(u+1)}
\end{equation} 
where $q_1 = (1-\beta_{1,\max})(1-\epsilon)$ and $q_2 = (1+\beta_{2,\max})(1+\epsilon)$. Similarly each $\tilde{\alpha_i}$ can be bounded using $\alpha_i$ as $(1-\theta_{1,i})\alpha_i \leq \tilde{\alpha_i} \leq (1+\theta_{2,i})\alpha_i$ where 
\begin{equation*}
\theta_{k,i} = \frac{\sum_{j=1}^N\delta_{k,i,j} b_{i,j}}{\sum_{j=1}^N b_{i,j}},
\end{equation*}
$ k \in \{1,2\}$ using which $\tilde{\alpha}$ can be bounded in terms of $\alpha$ as 
\begin{equation}
\label{boundForAlpha}
(1-p_1)\alpha \leq \tilde{\alpha} \leq (1+p_2)\alpha
\end{equation}
where 
\begin{equation*}
p_k = \frac{\sum_{i=1}^{N-K-2} \theta_{k,i} \alpha_i (NC_{\max}+1)^{i-1}}{\alpha}.
\end{equation*}
From the equations (\ref{boundForSum}) and (\ref{boundForAlpha}), it is clear that for the permutation $\pi$ that satisfies (\ref{permProblem}),
\begin{equation}
\label{necessaryCondition}
\frac{1-p_1}{q_1} \alpha \leq \sum_{u=1}^{N} h_{\pi(u)}x_{\pi(u+1)} \leq \frac{1+p_2}{q_2} \alpha,
\end{equation}
which gives us the necessary condition for the given set $S$ to be $\epsilon$ K-compressible with the $\epsilon$-perturbation of $x_N$ assumed as the last element in the orbit of a $K^{th}$ degree polynomial. It is worth emphasizing that all the constants in the above equation (\ref{necessaryCondition}) can be directly computed from the given set $S$. Repeating this process for the $N$ different choices of $x_N$ gives us $N$ different conditions similar in form to equation (\ref{necessaryCondition}) with the constants taking on different values for different constraints. Then the necessary condition for the set to be $\epsilon$ K-compressible is the existence of a permutation $\pi$ that satisfies at least one of these $N$ constraints.

Given a set $Z = \left\{\left(\begin{matrix}h_1\\x_1\end{matrix}\right),\left(\begin{matrix}h_2\\x_2\end{matrix}\right),\ldots,\left(\begin{matrix}h_N\\x_N\end{matrix}\right)\right\}$ of $N$ tuples and constants $c_1$ and $c_2$, the problem of finding the permutation $\pi$ such that
\begin{equation*}
c_1 \leq \sum_{u=1}^{N} h_{\pi(u)}x_{\pi(u+1)} \leq c_2
\end{equation*}
can be transformed to the multi-criteria traveling salesman problem as discussed under section (\ref{Trans2MultiTSP}) and the solutions given in \cite{Blaser08,Manthey09a,Manthey09b} can be adopted to solve our problem.
\section{Conclusion and open problems}
\label{conclusion}
We introduced a new notion of compressibility of sets of numbers where the set is represented through the coefficients of smaller $K^{th}$ degree polynomial $f$ and is produced by repeated composition of $f$ with itself. How effective the compression is depends upon the inherent structure the numbers constituting the set exhibits. Smaller the degree of the polynomial ($K$) w.r.t the size of the set ($N$), higher the compression ratio. We then provided approaches to determine whether a given set is K-compressible by transforming it to the  multi-criteria traveling salesman problem (TSP). The solutions developed for the multi-criteria TSP can then be availed to obtain solutions for our problem. Though we didn't formally prove it, we showed why it is unlikely for the K-compressibility problem to have a polynomial time solution. We then discussed about the notion of $\epsilon$ K-compressibility for the case of lossy compression, where we provided the necessary condition that the given set should satisfy for it to be $\epsilon$ K-compressible, which can be tested using the algorithms developed for multi-criteria TSP.

A random set of $N$ elements may not be K-compressible for all values of $K < N-2$. Hence trying to reproduce the elements of a set by successive composition of smaller degree polynomial function might have limited practical relevance in its current form. Nevertheless, our work on the current problem does provide headway to solve many generalizations of it which might have widespread applicability. These open problems are listed in the subsequent section and we would like to address them in our future work. We also pose some open problems in dynamical system theory which are interesting from a core theoretical perspective.

\subsection{List of open problems}
\begin{enumerate}
\item Is K-compressibility NP-hard? 
\item Given a set $S$ of $N$ numbers and integers $m$ and $K$, can $S$ be partitioned into $m$ subsets such that each subset is K-compressible with or without the same compressing function? The current work provides solution for the case $m=1$.
\item Can K-compressibility be extended for other class of functions like piece-wise linear functions or even splines? Does this class of function exhibit similarity to the problem defined above? Is \emph{``partitioning''} the function equivalent to partitioning the set?
\item Given a set $S$ of $N$ numbers and integers $m$ and $K$, can $S$ be made K-compressible by addition of utmost $m$ numbers to it?
\item For every $K \geq 2$, does $\exists$ integer $M(K)$, such that $\forall N \geq M(K)$, any K-compressible set $S$ consisting of $N$ distinct numbers is uniquely K-compressible, i.e there exist utmost one compressing function $f \in \prod_K, f \notin \prod_1$?.
\item For every $K$, does $\exists$ an integer $U(K)$, such that for any set $S$ with cardinality $N \geq U(K)$, the number of bijective mappings from $S$ to $S$ defined using $f \in \prod_K$ is utmost polynomial in $N$ and $K$? If this is true, then enforcing the Hamiltonian path constraint--defined in section(\ref{EqvTrans})--may be unnecessary if $N \geq U(K)$.
\end{enumerate}
\bibliographystyle{elsarticle-num}
\bibliography{DynamicsCompression}

\end{document}